%% file: CuellC-PatrickGW-2007-1.tex
\documentclass{article}\input{defs.tex}

\allowdisplaybreaks
\newcommand{\deltaf}{\delta\-3f}
\begin{document}

\title{Skew critical problems}

\author{
Charles Cuell and George W.\ Patrick\\[.05in]
\small Applied Mathematics and Mathematical Physics\\[-.05in]
\small Department of Mathematics and Statistics\\[-.05in]
\small University of Saskatchewan\\[-.05in]
\small Saskatoon, Saskatchewan, S7N~5E6, Canada
}

\date{\small March 2007$^\dag$}

\maketitle

\renewcommand{\thefootnote}{}
\footnotetext{$^\dag\backslash\mbox{today}$: \today}

\vspace*{-.3in}\begin{abstract} Skew critical problems occur in
continuous and discrete nonholonomic Lagrangian systems. They are
analogues of constrained optimization problems, where the objective is
differentiated in directions given by an apriori distribution, instead
of tangent directions to the constraint. We show semiglobal existence
and uniqueness for nondegenerate skew critical problems, and
show that the solutions of two skew critical problems have the
same contact as the problems themselves. Also, we develop some
infrastructure that is necessary to compute with contact order
geometrically, directly on manifolds.
\end{abstract}

%
\section{Introduction}
%
%
Let $M$ and $N$ be manifolds, suppose $f\colon M\to\RR$ is $C^1$, and
let $g\colon M\to N$ be a $C^1$ submersion. Given this data,
$m_c\in M$~is a \defemph{critical point} at~$n\in N$ if
\[[eq:ordinary_critical_problem]
\left\{\begin{array}{l}
  \mbox{$\displaystyle df(m_c)(v)=0$ for all $\displaystyle v$ 
    such that $\displaystyle T_{m_c}g(v)=0$,}\\[2pt]
  \displaystyle g(m_c)=n.
\end{array}\right.
\]
This is the standard constrained optimization problem that seeks
critical points of the objective~$f$ subject to the constraint~$g$.

Appearing in~\eqref{eq:ordinary_critical_problem} are the derivative
of the objective~$df$, the constraint function~$g$, and $\ker Tg$,
which is a distribution on $M$. Generalizing, we consider the data
$(\alpha,\kD,g)$, where $\alpha$ is a one-form on~$M$, $\kD$ is a
distribution on~$M$, and $g\colon M\to N$ is as above. We replace the
first condition of~\eqref{eq:ordinary_critical_problem} with the
condition that $\alpha$ annihilates~$\kD$, and we call the result a
\defemph{skew critical problem}. Skew critical problems occur when an
objective function is not differentiated in tangent directions to a
constraint, but rather is differentiated in the directions specified
by an apriori given distribution.  We are interested in skew critical
problems because, for \emph{nonholonomic mechanics, the relevant
variational principle is skew}~\cite{PatrickGW-2006-1}, and this is
also true of the variational discrete analogues of nonholonomic
systems.

For mechanics we are interested in existence and uniqueness of skew
critical problems, by direct perturbation from the point of zero-time
change. We have a global solution of the (trivial) zero-time problem,
and we are interested in \defemph{semiglobal} results, which means
global along the unperturbed problem, but local transverse to that.
For discrete nonholonomic systems, we are also interested to know that
the solutions of two skew critical problems have the same contact as
the data of the two problems. The skew critical problems of mechanics
require desingularization at zero-time, essentially by dividing by
time. This degrades the order matching, which is again recovered by a
zero-time symmetry of the desingularized problem, and so we must
consider the presence of symmetry.  We are interested in applications
to both the continuous and discrete mechanics, so we work in an
appropriate context of infinite dimensional manifolds.

In this work, we collect some technical results related to skew
critical problems. For such a problem~$(\alpha,\kD,g)$, little can be
inferred just from the equations $\alpha(m)|\kD=0$, $g(m)=n$, without
some control imposed on~$\alpha$,~$\kD$, and~$g$, so we begin in
Section~\ref{section:skew-critical-problems} with the definition of a
\defemph{nondegenerate} skew critical point. This corresponds to
infinitesimal conditions that, using the implicit function theorem,
imply there is locally a unique skew critical point for every nearby
constraint value (Lemma~\ref{lem:skew-critical-point-local}). If~$N$
is paracompact, then a manifold of nondegenerate skew critical points
along a submanifold~$N_0\subseteq N$ can be extended along the whole
of~$N_0$. We call this result \defemph{semiglobal} because it
establishes an extension over the whole of~$N_0$, rather that just at
one point of~$N$.

Contact of solutions of skew critical problems is important for
discretizations of constrained Lagrangian systems, because contact
with the exact system determines the order of the corresponding
numerical methods. Section~\ref{section:order} establishes the basic
definitions and results about contact. Generally, it often happens
that cancellations result in one higher contact that would normally be
expected from data or computation. For example, any Taylor expansion
to odd order of an even function, is actually the expansion to the
next higher order; a less trivial example is the fact that any odd
order self-adjoint one step numerical method is one higher (even)
order~\cite{HairerE-LubichC-WannerG-2002-1}. It is best to understand
the cancellations geometrically. This kind of ``passage to the next
order'' occurs when a geometric object that we call the
\defemph{residual} vanishes.  In Section~\ref{section:order} we find
that it is useful to consider the vector bundle analogue of blowing up
near the zero of a function of a single variable i.e.\ the function
$f(t)/t$ where $f(0)=0$. The completion of the function is made with
the help of the vertical bundle at the zero section, and the contact
drops by one. We provide, for computing on manifolds,
Equation~\eqref{eq:g_circ_f_res}, which computes the residuals of the
composition of two maps in terms of the residuals of the maps
themselves. For skew critical problems, it is necessary to consider
the contact order of \emph{distributions}, which are subsets rather
that maps. This is naturally done using Grassmann manifolds: a
distribution can be regarded as an assignment of subspaces to base
points.

Finally, in Section~\ref{section:equations} we consider contact for
inverse functions and the problems of construction maps from
graphs. For graphs, an exchange symmetry of the residuals implies that
the contact increases by one. In Section~\ref{section:skew-problems}
we consider contact for skew critical problems. In the presence of the
action of a Lie group, we obtain equivariance of the residuals of the
skew critical points given equivariance of the residuals of the skew
critical problems.

The notations in this work follow those
of~\cite{AbrahamR-MarsdenJE-RatiuTS-1988-1}. We assume without mention
that the manifolds and submanifolds we use are sufficiently
differentiable to support whatever operations are involved.

%
\section{Regular skew critical problems}
\label{section:skew-critical-problems}
%
%
Let $M$ and $N$ be Banach manifolds, $\alpha$ be a $C^k$ one-form on
$M$, $\kD$ be a $C^k$ distribution on $M$, and let $g\colon M\to N$ be
a $C^k$ submersion i.e.\ $Tg$ is surjective with split kernel. We call
$(\alpha,\kD,g)$ a \defemph{$C^k$ skew critical problem}.

\begin{definition}
A point $m_c\in M$ is a \defemph{skew critical point of $(\alpha,\kD,g)$}
at $n\in N$ if\[[eq:skew_critical_problem]
\left\{\begin{array}{l}
  \mbox{$\displaystyle \alpha(m_c)(v)=0$ for all 
    $\displaystyle v\in\kD_{m_c}$,}\\[2pt]
  \displaystyle g(m_c)=n.
\end{array}\right.
\]
\end{definition}

A critical point $m_c$ of a constrained optimization problem with
$n\equiv g(m_c)$ is called \defemph{nondegenerate} if the Hessian of
$f|g^{-1}(n)$ is nonsingular.  The corresponding notion for skew
critical problems is given below in Definitions~\ref{def:skew-Hessian}
and~\ref{def:skew-regular}.

\begin{definition}\label{def:skew-Hessian}
Let $m_c$ be a skew critical point of $(\alpha,\kD,g)$.  Define the
bilinear form $d_\kD\alpha(m_c)\colon T_{m_c}M\times\kD_{m_c}\to\RR$
by
\[
d_\kD\alpha(m_c)(u,v)\equiv\bigl\langle d(i_V\alpha)(m_c),u\bigr\rangle,
\]
where $V$ is a (local) vector field with values in $\kD$ such that
$V(m_c)=v$.  The \defemph{skew Hessian} of $\alpha$ with respect to
$g$ and $\kD$ is the bilinear form
\[
d_{\kD,g}\alpha(m_c)\colon\ker T_{m_c}g\times\kD_{m_c}\to\RR
\] 
obtained by restriction of $d_\kD\alpha(m_c)$. Define
$d_{\kD,g}\alpha(m_c)^\flat\colon\ker T_{m_c}g\to\kD_{m_c}^*$ by
\[
d_{\kD,g}\alpha(m_c)^\flat(u)\equiv d_{\kD,g}\alpha(m_c)(u,\,\cdot\,).
\]
\end{definition}

\begin{remark}\label{rem:Hessian-ok}
The definition of $d_\kD\alpha(m_c)$ does not depend on the
extension~$V$:~in a vector bundle chart of $\kD$, the local setup
has
\[
TM=U\times(\DD\oplus\FF),\quad 
\kD= U\times\bigl(\DD\oplus\sset0\bigr),\quad
\alpha=\alpha_\DD\oplus\alpha_\FF,
\]
where $U\subseteq\DD\oplus\FF$ is open, $\alpha_\DD\colon
U\to\DD^*\cong\ann\FF$, and $\alpha_\FF\colon
U\to\FF^*\cong\ann\DD$. Supposing that $x_c\in U$ is a skew critical
point, two extensions $V_i\colon U\to\DD$, $i=1,2$, with
$V_1(x_c)=v=V_2(x_c)$ result in
$i_{V_1-V_2}\alpha=\langle\alpha_\DD,V_1-V_2\rangle$.  By the product
rule, the derivative of this at $x_c$ is zero since both $\alpha_\DD$
and $V_1-V_2$ vanish at $x_c$, so
$d(i_{V_1}\alpha)(x_c)=d(i_{V_2}\alpha)(x_c)$. In contrast to the
constrained critical problems, skew Hessians are not symmetric since
their arguments assume values in different vector subspaces.
\end{remark}

\begin{definition}\label{def:skew-regular}
A skew critical point $m_c$ of $(\alpha,\kD,g)$ is called
\defemph{nondegenerate} if $d_{\kD,g}\alpha(m_c)^\flat$ is a linear
isomorphism.
\end{definition}

In finite dimensions, the standard constrained optimization
problem~\eqref{eq:ordinary_critical_problem} has as many equations for
$m_c$ as there are unknowns, because $g$ simultaneously constrains
both~$v$ and~$m_c$. For the skew
problem~\eqref{eq:skew_critical_problem}, the number of equations need
not equal the number of unknowns, since $g$ and $\kD$ may be unrelated.
Definition~\ref{def:skew-regular} controls this, because if $m_c$ is
nondegenerate then the fiber dimensions of $\ker T_{m_c}g$ and
$\kD_{m_c}$ are equal since $\ker T_{m_c}g$ and $\kD_{m_c}^*$ are
isomorphic.

\begin{lemma}\label{lem:skew-critical-point-local}
Let $m_c$ be a nondegenerate skew critical point of a $C^k$ skew
critical problem $(\alpha,\kD,g)$, $k\ge 1$, and let $n_c\equiv
g(m_c)$.  Then there are neighborhoods $U\ni m_c$ and $V\ni n_c$ such
that, for every $n\in V$ there is a unique skew critical point~$m\in
U$ of $(\alpha,\kD,g)$ such that $g(m)=n$. Moreover, the map
$\gamma\colon V\to U$ so defined is $C^k$.
\end{lemma}

\begin{proof}
Using vector bundle charts as in Remark~\ref{rem:Hessian-ok}, the skew
critical points~$x$ such that $g(x)=y$ are obtained by solving
$F(x)=(0,y)$, where $F(x)\equiv\bigl(\alpha_\DD(x),g(x)\bigr)$. The
derivative of $F$ at a particular $x_c$ is
\[[eq:DF]
DF(x_c)u=\bigl(D\alpha_\DD(x_c)u,Dg(x_c)u\bigr).
\]
The first component is a linear isomorphism on $\ker Dg(x_c)$ since
$x_c$ is nondegenerate. Since $Dg(x_c)$ is onto with a split kernel, there
is a closed subspace $\KK$ such that $\DD\oplus\FF=\ker
Dg(x_c)\oplus\KK$, and $Dg(x_c)|\KK$ is a linear
isomorphism. From~\eqref{eq:DF}, $DF(x_c)u=(w_1,w_2)$ is continuously
inverted by
\[
&\tilde u=(Dg(x_c)|\KK)^{-1}w_2,\\
&u=\tilde u+\bigl(D\alpha_\DD(x_c)|\ker Dg(x_c)\bigr)^{-1}\bigl(w_1-D
  \alpha_\DD(x_c)\tilde u\bigr),
\]
and the result follows from the inverse function theorem.
\end{proof}

The following semiglobal inverse function theorem is found on page~97
of~\cite{LangS-1972-1}. The semiglobal result for skew critical points
which follows that, the proof of which is included for completeness,
pre-supposes nondegeneracy along a given smooth map of skew critical
points.

\begin{theorem}\label{thm:semiglobal-inverse-fn}
Let $M$ and $N$ be manifolds and $f \colon M \to N$ be
$C^k$, $k \ge 1$.  Suppose that
\begin{enumerate}
\item $M_0$ is a closed submanifold of $M$, $N_0$ is a closed
submanifold of $N$, and $f|M_0\colon M_0\to N_0$ is a diffeomorphism; and
\item $f$ is a local diffeomorphism at every $m\in M_0$.
\end{enumerate}
Then $f$ is a $C^k$ diffeomorphism from some open neighborhood $U
\supset M_0$ to some open neighborhood $V\supset N_0$.
\end{theorem}

\begin{theorem}\label{theorem:semiglobal-skew-critical-point} 
Let $(\alpha,\kD,g)$ be a $C^k$ skew critical problem, $k\ge1$, where
$g\colon M\to N$.  Suppose that $N$ is paracompact, and that
{\renewcommand{\theenumi}{\alph{enumi}}\begin{enumerate}
\item $M_0$ is a closed submanifold of $M$, $N_0$ is a closed manifold
of $N$ and $\gamma_0\colon N_0\to M_0$ is a $C^k$
diffeomorphism; and
\item for all $n \in N_0$, $\gamma_0(n)$ is a nondegenerate skew critical
point of  $(\alpha,\kD,g)$ at~$n$.
\end{enumerate}
Then there are open neighborhoods $U\supseteq M_0$ and $V\supseteq
N_0$ and a $C^k$ extension $\gamma\colon V\to U$ such that}
\begin{enumerate}
\item 
for all $n\in V$, $\gamma(n)$ is a skew critical point of
$(\alpha,\kD,g)$ at~$n$; and
\item
$\gamma(n)$ is the unique skew critical point of $(\alpha,\kD,g)$ in
$U$.
\end{enumerate}
\end{theorem}

\begin{proof}
Applying Lemma~\ref{lem:skew-critical-point-local} at
all~$\gamma_0(n_0)$ as $n_0$ ranges through $N_0$, there are open
covers $U_i$ of $M_0$ and $V_i$ of $N_0$, and $C^k$~maps
$\gamma_i\colon V_i\to U_i$, such that, for all~$n\in V_i$,
$\gamma_i(n)$ is the unique skew critical point of~$(\alpha,\kD,g)$
in~$U_i$. By shrinking~$V_i$ one can arrange $\gamma_i(\cl
V_i)\subseteq U_i$ where $\gamma_i$ is defined on an open superset
of~$V_i$.  Because $N$ is paracompact, its open cover
$\sset{N\setminus N_0, V_i}$ admits a locally finite refinement, so
the collection $\sset{V_i}$ can be assumed locally finite.

By Lemma~20.4 of \cite{WillardS-1970-1}, the collection $\sset{\cl
V_i}$ is also locally finite, so each~$n\in\bigcup_iV_i$ admits a
neighborhood $V_n$ that meets only finitely~many~$\cl V_i$. For
each~$n\in\bigcup_iV_i$, the set of indices
\[
\St(n)\equiv\bigset{i}{n\in\cl V_i}
\]
is finite. No $\St (n)$ is empty because every $n\in\bigcup_i V_i$ is
contained in some $ V_i$ and hence is in some $\cl V_i$. The set
\[[20]
V_n\setminus
\bigcup\set{\cl V_i}{\mbox{$\cl V_i$ 
meets $V_n$ and $i\not\in \St (n)$}}
\]
an open neighborhood of $n$ because it subtracts from $V_n$ only
finite many closed sets, and it has the property that if any of its
members is in any $\cl V_i$ then $i\in \St (n)$. Replacing each $V_n$
with \eqref{20}, it can be assumed that $\St (n^\prime)\subseteq \St
(n)$ for all $n^\prime\in V_n$.

Defining
\[
U\equiv\bigcup_{n\in\bigcup_i V_i}\biggl(g^{-1}(n)\cap
\bigcap_{i\in \St (n)} U_i\biggr),
\qquad V\equiv\bigcup_i\gamma_i^{-1}(U),
\]
we can show the following facts.
\begin{enumerate}
\item 
$M_0\subseteq U$: if $m\in M_0$ and $n\equiv g(m)$ then $m\in g^{-1}(n)$ and
$n\in\cl V_i$ for all $i\in\St (n)$ so $\gamma_i(n)\in\gamma_i(\cl
V_i)\subseteq U_i$ for all $i\in\St (n)$, hence $m\in U$.
\item 
$U$ is an open neighborhood of $M_0$: if $m\in U$ and $n\equiv g(m)$ then
\[
m\in g^{-1}(n)\cap\bigcap_{i\in \St (n)} U_i
\subseteq g^{-1}(V_n)\cap\bigcap_{i\in \St (n)} U_i.
\]
The last set is open because it is the intersection of finitely many
open sets. Also,
\[
g^{-1}(V_n)\cap\bigcap_{i\in \St (n)} U_i
&=\bigcup_{n^\prime\in V_n}\biggl(g^{-1}(n^\prime)\cap
\bigcap_{i\in \St (n)} U_i\biggr)\\
&\subseteq\bigcup_{n^\prime\in V_n}\biggl(g^{-1}(n^\prime)\cap
\bigcap_{i\in \St (n^\prime)}U_i\biggr)\\
&\subseteq U.
\] 
Thus there is an open neighborhood of $m$ that is contained in $U$.
\item
$U$ has the property that, for all $m_1,m_2\in U$, $g(m_1)=g(m_2)$
implies that there is an $i$ such that $m_1$ and $m_2$ are both in
$U_i$. Indeed, any such $m_1$ and $m_2$ are members of
\[
g^{-1}(n)\cap\bigcap_{i\in\St (n)}U_i,
\]
where $n=g(m_1)=g(m_2)$, and so both $m_1$ and $m_2$ are members of
any $U_i$ for any $i\in\St (n)$.
\end{enumerate}

Let $n\in V$. Then $n\in\gamma_i^{-1}(U)$ for some $i$ and
$m=\gamma_i(n)$ is a skew critical point of $(\alpha,\kD,g)$ in
$U$. If $m^\prime\in U$ is another such skew critical point then
$g(m)=g(m^\prime)$, and $m$ and $m^\prime$ both lie in a single
$U_j$. By definition of the $U_j$ there is only one skew critical
point of $(\alpha,\kD,g)$ in $U_j$, so $m=m^\prime$.  Thus for all $n\in
V$ there is a unique skew critical point of $(\alpha,\kD,g)$ in
$U$. Define $\gamma\colon V\to U$ by this correspondence.  By the
uniqueness used to define $\gamma$, the restriction of $\gamma$ to any
$\gamma_i^{-1}(U)$ is $\gamma_i$, which is $C^k$, and the
$\gamma_i^{-1}(U)$ cover $V$, so $\gamma$ is $C^k$.
\end{proof}

%
\section{Order Notation and Residuals}
\label{section:order}
%
%

Given two functions~$f_i(x)$, $i=1,2$, of a single variable~$x\in\RR$,
the standard definition of $f_1(x)=f_2(x)+O(x^r)$ is that there are
numbers~$\delta>0$ and~$C>0$ such that $|f_1(x)-f_2(x)|\le C|x|^r$
for~$|x|<\delta$.  If the functions~$f_i$ are $C^r$, $r\ge1$, then
$f_1(x)=f_2(x)+O(x^r)$ if and only if there is a continuous function,
say $\deltaf(x)$, such that $f_1(x)=f_2(x)+x^r\deltaf(x)$.  The
following definitions export the second formulation to the context of
manifolds.

{\samepage \begin{definition}\mbox{}
\begin{enumerate}
\item 
Let~$M$ be a manifold and $h_M\colon M\to\mathbb{R}$ be a
$C^\infty$~function which has~$0$ as a regular value. The pair
$(M,h_M)$ will be called a \defemph{manifold}.
\item 
Let $(M,h_M)$ and $(N,h_N)$ be manifolds. A \defemph{$C^k$~mapping}
$f\colon(M,h_M)\to(N,h_N)$ is a $C^k$~mapping $f\colon M\to N$ such
that $h_N\circ f=h_M$.
\item
A \defemph{$C^k$~mapping} $f\colon(M,h_M)\to N$ or $f\colon M\to
(N,h_N)$ is a mapping $f\colon M\to N$ without any conditions
involving $h_M$ or $h_N$.
\end{enumerate}\end{definition}}

\begin{definition}\label{def:order}
Let~$(M,h_M)$ and~$N$ be manifolds and $f_i\colon (M,h_M)\to N$,
$i=1,2$, be such that $f_1=f_2$ on $h_M^{-1}(0)$. Define
$f_2=f_1+O(h_M^r)$, $r\ge 1$ if, for all $m_0\in h_M^{-1}(0)$, there
is a chart~$\nu$ at $n_0\equiv f_i(m_0)\in N$, and there is a
function $(\deltaf)_{\nu}$ defined near~$m_0$, and continuous
at~$m_0$, such that
\[
\nu\bigl(f_2(m)\bigr)-\nu\bigl(f_1(m)\bigr)=h_M(m)^r(\deltaf)_{\nu}(m),
\]
for all $m$ in some neighborhood of $m_0$.
\end{definition}

The definition of $f_2=f_1+O(h_M^r)$ does not depend on the
coordinate chart:~if~$\nu$ and~$\tilde{\nu}$ are two
coordinate charts at~$n_0$, as in Definition~\ref{def:order}, and
for~$m$ near to~$m_0$,
\[
\tilde\nu&\bigl(f_2(m)\bigr)-\tilde\nu\bigl(f_1(m)\bigr)\\
&=\bigl(\tilde\nu\circ\nu^{-1}\bigr)\bigl(\nu\bigl(f_2(m)\bigr)\bigr)
  -\bigl(\tilde\nu\circ\nu^{-1})\bigl(\nu\bigl(f_1(m)\bigr)\bigr)\\
&=\bigl(\tilde\nu\circ\nu^{-1}\bigr)\Bigl(\nu\bigl(f_1(m)\bigr) 
  +h_M(m)^r(\deltaf)_\nu(m)\Bigr)
  -\bigl(\tilde\nu\circ\nu^{-1}\bigr)\bigl(\nu\bigl(f_1(m)\bigr)\bigr)\\
&=\int_0^1\frac{d}{dt}\bigl(\tilde\nu\circ\nu^{-1}\bigr)
  \Bigl(\nu\bigl(f_1(m)\bigr)+t\,h_M(m)^r(\deltaf)_\nu(m)\Bigr)\,dt\\
&=h_M(m)^r\left[\int_0^1 D\bigl(\tilde\nu\circ\nu^{-1}\bigr)
  \Bigl(\nu(f_1(m))+t\,
  h_M(m)^r(\deltaf)_\nu(m)\Bigr)\,dt\right]\,(\deltaf)_\nu(m),
\]
as required.

The quantities $(\deltaf)_\nu(m_0)$ and $(\deltaf)_{\tilde
\nu}(m_0)$ transform as tangent vectors. Indeed,
\[
h_M(m)^r(\deltaf)_{\tilde\nu}(m)=\tilde\nu\bigl(f_2(m)\bigr)
  -\tilde\nu\bigl(f_1(m)\bigr),
\]
so 
\[
(\deltaf)_{\tilde\nu}(m)=\left[\int_0^1 D\bigl(\tilde\nu\circ\nu^{-1}\bigr)
  \Bigl(\nu(f_1(m))+t\,h_M(m)^r(\deltaf)_\nu(m)\Bigr)\,dt\right]
  (\deltaf)_\nu(m).
\]
At~$m_0\in h_M^{-1}(0)$, and setting~$n_0\equiv f_i(m_0)$,
\[
(\deltaf)_{\tilde\nu}(m_0)=D(\tilde\nu \circ \nu^{-1})\bigl(\nu(n_0)\bigr)
  (\deltaf)_\nu(m_0), 
\]
as required.

\begin{definition}
Let $(M,h_M)$ be a manifold, $f_2=f_1+O(h_M^r)$, and $m\in
h_M^{-1}(0)$. The vector $\res^r(f_2,f_1)(m)\in T_nN$ with
representation $(\deltaf)_\nu(m)$ for any chart~$\nu$ is called the
\defemph{$r$-residual of~$f_2$ with respect to~$f_1$}.
\end{definition}

The residual $\res^r(f_2,f_1)$ is defined only on $h_M^{-1}(0)\subset
M$ and takes values in~$TN$. The condition~$f_2=f_1+O(h_M^r)$, can be
localized to a point of~$M$ or a subset of~$M$ in the obvious way, and
the residual will be correspondingly localized. In general, jets of
mappings between manifolds carry an affine action by a geometrically
based vector space, amounting basically to the first nonzero term of
the Taylor series of the difference between two mappings. Also, the
notion of \defemph{contact} below is the same as the contact
equivalence in the definition of
jets~\cite{KolarI-MichorPW-SlovakJ-1993-1}.

If $(M,h_M)$ is a manifold, then, since $0$~is a regular value
of~$h_M$, there are \defemph{$h_M$-adapted charts} at each~$m_0\in
h_M^{-1}(0)$ i.e.\ charts such that the local representative of
$h_M$~is the projection~$(x,t)\mapsto t$. We can prove an equality or
formula concerning residuals in any chart since residuals are
geometric, and in particular, we can always use an $h_M$-adapted
chart.  

Suppose that $f_i\colon (U,h_U)\to V\subseteq\FF$,~$i=1,2$, are~$C^r$,
where $U$~is an open subset of~$\EE\times\RR$, where $\EE$ and
$\FF$~are a Banach spaces, and where $h_U(x,t)=t$. For fixed~$x$, the
Taylor expansions in~$t$ about~$t=0$ of the~$f_i$ are
\[
f_i(x,t)=f_i(x,0)+t\frac{\partial
f_i}{\partial t}(x,0)+\cdots+\frac{t^r}{r!}
\frac{\partial^rf_i}{\partial t^r}(x,0)+R_{r,i}(x,t)\,t^r,  
\]
where~\cite{AbrahamR-MarsdenJE-RatiuTS-1988-1}
\[
R_{r,i}(x,t)=\int_0^1\frac{(1-s)^{r-1}}{(r-1)!}\left(
\frac{\partial^rf_i}{\partial t^r}(x,st)-
\frac{\partial^rf_i}{\partial t^r}(x,0)\right)\,ds.
\]
The condition that~$f_2=f_1+O(h_U^r)$ at~$(x,0)$ is thus equivalent to
the condition that these Taylor expansions match at~$(x,0)$ up to and
including the degree~$r-1$ term. So, given this,
\[
f_2(x,t)&-f_1(x,t)\\
  &=\frac{t^r}{r!}\left(\frac{\partial^rf_2}{\partial t^r}
  (x,0)-\frac{\partial^rf_1}{\partial t^r}
  (x,0)\right)+R_{r,2}(x,t)\,t^r-R_{r,1}(x,t)\,t^r,
\]
which identifies $(\deltaf)_\nu(x,t)$ in these coordinates as
\[
(\deltaf)_\nu(x,t)=\frac1{r!}\left(
  \frac{\partial^rf_2}{\partial t^r}(x,0)
  -\frac{\partial^rf_1}{\partial t^r}(x,0)\right)+R_{r,2}(x,t)-R_{r,1}(x,t).
\]
Setting~$t=0$, the residual is
\[
  (\deltaf)_\nu(x,0)
  & = \frac1{r!}\left.\frac{\partial^r}{\partial t^r}\right|
  _{t=0}\bigl( f_2(x,t)-f_1(x,t) \bigr).
\]
If $\res^r(f_2,f_1)=0$, then the Taylor series of~$f_1$ and~$f_2$
agree up to and including terms of degree~$r$, one more than the
degree~$r-1$ agreement implied by~$f_2=f_1+O(h_U^r)$. Thus, if it is
necessary to establish with some computation, that two functions which
differ at order~$r$, actually differ at order~$r+1$, then one can
accomplish this by showing that $\res^r(f_2,f_1)=0$.

\begin{definition}
If $M$, $h_M$, and $f_i$ are as in Definition~\ref{def:order}, then
$f_1$ and~$f_2$ have order $h_M^{r-1}$~contact, or just have
contact~$r-1$, if $f_2=f_1+O(h_M^r)$.
\end{definition}

If $\pi\colon E\rightarrow M$ is a vector bundle then $\ker T\pi$ is a
subbundle of $TE$, called the \emph{vertical subbundle}. Recall that
the tangent space of the zero section defines a natural horizontal
subspace, so any vector of~$TE$ at the zero section can be split into
horizontal and vertical parts. This splitting can be defined by the
linear isomorphism
\[[120]
TM\oplus E\rightarrow TE:\quad(v_m,w_m)\mapsto
\left.\frac d{dt}\right|_{t=0}0_{m(t)}+\left.\frac d{dt}\right|_{t=0}t w_m,
\]
where $m(t)$ is a curve in $M$ such that $m^\prime(0)=v_m$.  If $z\in
T_{0_m}E$ then denote the horizontal and vertical parts of $z$ by
$\hor z\in T_mM$ and $\vrt z\in E_m$, respectively i.e.\ the inverse
of \eqref{120} is  $z\mapsto(\hor z,\vrt z)$.

If $f$ is a $C^1$ function such that $f(0)=0$, then it is
elementary that
\[
\hat f(t)\equiv\begin{cases}\displaystyle\frac{f(t)}t,&t\ne0,\\[5pt]
  \displaystyle f^\prime(0),&t=0,\end{cases}
\]
is continuous. The purpose of Lemma~\ref{prp:vb-div-by-zero} is to
show that a mapping on a manifold can be smoothly divided by a real
function that takes values in a vector bundle and is in the zero
section~$0(E)$ if the function vanishes.

\begin{proposition}\label{prp:vb-div-by-zero}
Let $(M,h_M)$ and $N$ be a manifolds, and let $\pi\colon E\rightarrow
N$ be a vector bundle. Suppose that $f\colon U\rightarrow E$ is~$C^k$,
$k\ge1$, and that $f(m)\in 0(E)$ whenever~$h_M(m)=0$. Then for all $m$
such that $h_M(m)=0$, there is a unique $e(m)\in E_{\pi(f(m))}$ such
that
\[[eq:blow_vertical_defn]
\vrt T_mf(v_m)=\bigl(dh_M(m)v_m\bigr)e(m),\qquad v_m\in T_mM.
\]
Moreover, the function $\hat f\colon M\rightarrow E$ defined by
\[
\hat f(m)\equiv\begin{cases}
\displaystyle\frac{f(m)}{h_M(m)},&h_M(m)\ne0,\\[7pt]
e(m),&h_M(m)=0,\end{cases}
\]
is $C^{k-1}$.
\end{proposition}

\begin{proof}
If $h_M(m)=0$ and $v_m\in\ker dh_M(m)$ then $v_m=c^\prime(0)$ for some
curve $c(t)\in h_M^{-1}(0)$. Since $f$ is in the zero section
wherever~$h_M$ is zero, it follows that $f\circ c(t)$ takes values in
the zero section, so $(f\circ h_M)^\prime(0)$ is horizontal.  Thus
$\vrt T_mf(v_m)=0$ for all~$v_m\in\ker d h_M(m)$, so there is a unique
$e(m)\in E_{\pi(f(m))}$ satisfying~\eqref{eq:blow_vertical_defn}.

We can set up an $h_M$-adapted chart~$\sset{(x,t)}$ on~$M$ and a
vector bundle chart on~$N$, so that $E=\sset{(y,e)}$, and
$f(x,t)=\bigl(f_0(x,t),f_1(x,t)\bigr)$.  Then $f_1(x,0)=0$ for all
$x$, so
\[[eq:vb-div-by-zero-1]
\vrt T_mf(x,0)(\delta x,\delta t)
=\frac{\partial f_1}{\partial x}(x,0)\delta x+
  \frac{\partial f_1}{\partial t}(x,0)\delta t
=\frac{\partial f_1}{\partial t}(x,0)\delta t
\]
and
\[[eq:vb-div-by-zero-2]
dh_M(x,t)(\delta x,\delta t)=\delta t.
\]
By comparison of~\eqref{eq:blow_vertical_defn}
with~\eqref{eq:vb-div-by-zero-1} and~\eqref{eq:vb-div-by-zero-2},
\[
e(x,0)=\frac{\partial f_1}{\partial t}(x,0),
\]
and it is required to show that $\hat f_1$ defined by
\[
\hat f_1(x,t)\equiv\begin{cases}\displaystyle\frac{f_1(x,t)} t,&
  t\ne0,\\[7pt]
\displaystyle\frac{\partial f_1}{\partial t}(x,0),&t=0,
\end{cases}
\]
is $C^{k-1}$.  At any $(x_0,0)$ the Taylor expansion of~$f_1$ is
\[[eq:prp:vb-div-by-zero_1]
f_1(x,t)&=Df_1(x_0,0)(\delta x,t)+\cdots\\
&\qquad\cdots+D^kf_1(x_0,0)(\delta x,t)^k+R_k(x,t)(\delta x,t)^k
\]
where $\delta x=x-x_0$ and $R(x_0,0)=0$. By differentiating
$f_1(x,0)=0$ in~$x$, $D^if_1(x_0,0)(\delta x,0)^i=0$ for $1\le i\le
k$, and substituting $t=0$ into \eqref{eq:prp:vb-div-by-zero_1} gives
$R(x,0)(\delta x,0)^k=0$. Thus the left side of
\eqref{eq:prp:vb-div-by-zero_1} has $t$ as a factor and
\[[eq:prp:vb-div-by-zero_2]
\hat f_1(x,t)&=\frac{\partial f_1}{\partial t}(x_0,0)
  +\frac1tD^2f_1(x_0,0)(\delta x,t)^2+\cdots\\
&\qquad\cdots+\frac1tD^kf_1(x_0,0)(t,h)^k+\frac1tR_k(x,t)(\delta x,t)^k.
\]
Each of the first~$k$ terms of~\eqref{eq:prp:vb-div-by-zero_2} are
polynomial in~$(\delta x,t)$ and the remainder is polynomial
in~$(\delta x,t)$ of degree~$k-1$ with coefficients functions
of~$(x,t)$ that vanish at~$(x_0,0)$. Thus, by the converse of Taylor's
theorem~\cite{AbrahamR-MarsdenJE-RatiuTS-1988-1}, $\hat f_1(x,t)$
is~$C^{k-1}$ at any~$(x_0,0)$.
\end{proof}

\begin{proposition}\label{prp:vb-div-by-zero-contact-order}
Let $(M,h_M)$ and $N$ be a manifolds, let $\pi\colon E\rightarrow N$ a
vector bundle, and suppose $f_i$ and $\hat f_i$ are as in
Proposition~\ref{prp:vb-div-by-zero}, with $k\ge r$. Then $\hat
f_2=\hat f_1+O(h_M^{r-1})$ if $f_2=f_1+O(h_M^r)$, $r\ge2$.  Moreover,
$\res^r(f_2,f_1)$ takes values in the vertical bundle of $E$ and
$\res^{r-1}(\hat f_2,\hat f_1)=\res^r(f_2,f_1)$.
\end{proposition}

\begin{proof}
Assume the context and notations of the proof
of~Proposition~\ref{prp:vb-div-by-zero}, so that
\[
f_1(x,t)=\bigl(f_{1,0}(x,t),f_{1,1}(x,t)\bigr),\qquad
f_2(x,t)=\bigl(f_{2,0}(x,t),f_{2,1}(x,t)\bigr).
\]
Since $f_{2,0}(x,t)=f_{1,0}(x,t)+O(t^r)$, the $r-1$~residual of the
first components of~$f_2$ and~$f_1$ is zero, and it suffices show that
\[[prp:vb-div-by-zero-contact-order_1]
\hat f_{2,1}(x,t)=\hat f_{1,1}(x,t)+t^{r-1}\deltaf(x,t)
\]
given $f_{2,1}=f_{1,1}+t^r\deltaf$, where $\deltaf$ is continuous
and
\[
\hat f_{i,1}(x,t)=\begin{cases}\displaystyle\frac{f_{i,1}(x,t)} t,&
  t\ne0,\\[7pt]
\displaystyle\frac{\partial f_{i,1}}{\partial t}(x,0),&t=0.
\end{cases}
\]
Equation~\eqref{prp:vb-div-by-zero-contact-order_1} can be shown in
the two cases $t=0$ and $t\ne0$: For~$t=0$,
\[
\frac{\partial f_{2,1}}{\partial t}(x,0)
  -\frac{\partial f_{1,1}}{\partial t}(x,0)
=\lim_{t\rightarrow 0}\frac{t^r\deltaf(x,t)}{t}
=\lim_{t\rightarrow 0}t^{r-1}\deltaf(x,t)=0,
\]
so even~$\hat f_{2,1}(x,t)=\hat f_{1,1}(x,t)$ in this case, whereas
for~$t\ne 0$,
\[
\hat f_{2,1}(x,t)=\frac{f_{2,1}(x,t)}t
 =\frac{f_{1,1}(x,t)+t^r\deltaf(x,t)}t
 =\hat f_{1,1}(x,t)+t^{r-1}\deltaf(x,t).
\]
\end{proof}

Proposition~\ref{prp:g_comp_f_order} is a key result because it can be
used to compute residuals without the invocation of local charts. Note
that if $(M,h_M)$ and $(N,h_N)$ are manifolds and $f\colon
M\rightarrow N$~is a $C^1$~function such that
$f\bigl(h_M^{-1}(0)\bigr)\subseteq h_N^{-1}(0)$, then for all~$m\in
h_M^{-1}(0)$, $d(h_N\circ f)(m)v_m=0$ for all~$v_m$ such that
$dh_M(m)v_m=0$. So one can define $\dot f\colon
h_M^{-1}(0)\rightarrow\RR$ by
\[
d(h_N\circ f)(m)=\dot f(m)\,dh_M(m).
\]
This is an instance of Proposition~\ref{prp:vb-div-by-zero} and it
follows that $\hat h_{N,f}$~defined by extending~$d(h_N\circ f)/h_M$
to~$\dot f$ on~$h_M^{-1}(0)$ is continuous.

\begin{proposition}\label{prp:g_comp_f_order}
Let $(M,h_M)$, $(N,h_N)$, and $P$ be manifolds, and suppose $f_i\colon
M\to N$ and $g_i\colon N\to P$, $i=1,2$ are $C^1$ and satisfy
$f_i\bigl(h_M^{-1}(0)\bigr)\subseteq h_N^{-1}(0)$, $f_2=f_1+O(h_M^r)$,
and $g_2=g_1+O(h_N^r)$.  Then $g_2\circ f_2=g_1\circ f_1+O(h_M^r)$.
Moreover, if~$h_M(m)=0$ and $n\equiv f_i(m)$, then
\[[eq:g_circ_f_res]
\res^r(g_2\circ f_2, g_1\circ f_1)(m)=\dot f_2(m)^r\,\res^r(g_2,g_1)(n)
  +T_ng_1\,\res^r(f_2,f_1)(m).
\]
\end{proposition}

\begin{proof}
It suffices to consider the local setup where $\EE$, $\FF$, and
$\GG$ are Banach spaces, $U\subseteq\EE$ and $V\subset\FF$ are open,
$f\colon U\to V$, $g\colon V\to W$, $h_U\colon U\to\RR$, and
$h_V\colon V\to\RR$. Then
\[ 
(g_2\circ f_2)(x)&=g_1\bigl(f_1(x)+h_U(x)^r\deltaf(x)\bigr)
  +h_V\bigl(f_2(x)\bigr)^r\delta g\bigl(f_2(x)\bigr)\\
&=g_1\bigl(f_1(x)\bigr)+\int_0^1\frac d{ds}\,g_1\bigl(
  f_1(x)+sh_U(x)^r\deltaf(x)\bigr)\,ds\\
&\qquad\mbox{}+h_V\bigl(f_2(x)\bigr)^r\delta g(f_2(x)\bigr)\\
&=g_1\bigl(f_1(x)\bigr)
+h_U(x)^r\left[\int_0^1
  Dg_1\bigl(f_1(x)+s\,h_U(x)^rf(x)\bigr)\,ds\right]
  \deltaf(x)\\
&\mbox\qquad\mbox{}+h_U(x)^r\hat h_{V,f_2}(x)^r\delta g\bigl(f_2(x)\bigr). 
\]
Assuming $x$ satisfies $h_U(x)=0$ and putting $y\equiv f_i(x)$,
results in 
\[ 
\res^r(g_2\circ f_2,g_1\circ f_1)(x)=Dg_1\bigl(y)
  \deltaf(x)+\hat h_{V,f}(x)^r\delta g(y\bigr),
\]
which is the local form of~\eqref{eq:g_circ_f_res}.

\end{proof}

\begin{remark}
If $r\ge2$, then $T_ng_1$ and $\dot f_2$ can be replaced by $T_ng_2$
and $\dot f_1$ respectively in Equation~\eqref{eq:g_circ_f_res}.  If
$h_N\circ f_i=h_M$ then $\dot f_i=1$ in any case.  Also, if $g_1=g_2$,
one can dispense with~$h_N$ and the assumption that
$f_i\bigl(h_M^{-1}(0)\bigr)\subseteq h_N^{-1}(0)$, obtaining the
formula
\[
\res^r(g\circ f_2, g\circ f_1)(m)=T_ng\,\res^r(f_2,f_1)(m).
\]
\end{remark}

Let $\pi\colon E\to B$ be a vector bundle with typical fiber~$\EE$,
and let $\EE_0$~be a closed split subspace of~$\EE$.  We will have use
of the \defemph{$\EE_0$-Grassmann bundle of~$E$, denoted
$\pi^{G(\EE_0,E)}_M\colon G(\EE_0,E)\to M$}, by which we mean the set
of subspaces of the fibers of~$E$ that are linearly isomorphic
to~$\EE_0$ i.e.\ the coset space of continuous linear injections (with
closed split image) of~$\EE_0$ into the fibers of~$E$ by the action
of~$\GL(\EE_0)$. The projection $\pi^{G(\EE_0,E)}_M$ associates
subspaces of the fiber~$E_m$ to~$m$, and the typical fiber
of~$G(\EE_0,E)$ is the Grassmann manifold~$G_{\EE_0}(\EE)$.  For more
information on Grassmann manifolds in the Banach space context,
see~\cite{AbrahamR-MarsdenJE-RatiuTS-1988-1}.

\begin{remark}
If $E^\prime$ is a $C^r$~subbundle of~$E$, with typical
fiber~$\EE^\prime$, then there is defined the $C^r$~map
$\iota_{E^\prime}\colon M\to G(\EE^\prime,E)$ that assigns to
any~$m\in M$ the subspace~$\iota_{\E^\prime}(m)=E^\prime_m$.
\end{remark}

\begin{remark}\label{rmk:grassmann-subspace}
As is well known, the tangent space at~$\BB\in G_{\EE_0}(\EE)$ is
canonically $\hom(\BB,\EE/\BB)$. Indeed, if $\BB(t)$ is a $C^1$ curve
in $G_{\EE_0}(\EE)$ with $\BB(0)=\BB$, then choose a splitting
$\EE=\BB\oplus\FF$ and define $\dot\BB\colon\BB\to\FF$ by
\[
\dot\BB\equiv\left.\frac{d}{dt}\right|_{t=0}
  \pi^{\EE}_{\EE/\BB}\circ\Bigl(\pi^{\BB\oplus
  \FF}_{\BB}\Bigm|\BB(t)\Bigr)^{-1},
\]
where $\pi^{\EE}_{\EE/\BB}$ denotes the projection of~$\EE$ to the
quotient $\EE/\BB$, and $\pi^{\BB\oplus\FF}_\BB$ denotes the
projection to~$\BB$ using the decomposition $\EE=\BB\oplus\FF$.
One verifies that $\dot\BB$ is independent of the choice of the
complement~$\FF$ and depends only on the tangent vector of~$\BB(t)$ at
$t=0$.
\end{remark}

\begin{definition}\label{def:D2_D1_order}
Let $(M,h_M)$ be a manifold and let $\kD_i$, $i=1,2$ be distributions
on $M$ such that $\kD_1|h_M^{-1}(0)=\kD_2|h_M^{-1}(0)$.  Define
$\kD_2=\kD_1+O(h_M^r)$ if $\iota_{\DD_2}=\iota_{\DD_1}+O(h_M^r)$, and
define $\res^r(\kD_2,\kD_1)\equiv\res^r(\iota_{\DD_2},\iota_{\DD_1})$.
\end{definition}

In the context of Definition~\ref{def:D2_D1_order}, note that the
assignment of the fibers of subbundles into the Grassmann bundle
preserves fibers, so $\pi^{G_\DD(TM)}_M\circ\iota_{D_i}=\!1_M$, and
\[
T\pi^{G_\DD(TM)}_M\res\bigl(\iota_{\kD_2},\iota_{\kD_1}\bigr)
&=\res\bigl(\pi^{G_\DD(TM)}_M\circ\iota_{D_2},
  \tau^{G_\DD(TM)}_M\circ\iota_{\kD_1}\bigr)\\
&=\res(\!1_M,\!1_M)\\
&=0,
\]
which shows that, for all $m\in h_M^{-1}(0)$, $\res^r(\kD_2,\kD_1)(m)$
is a vertical vector in $T\bigl(G_\DD(TM)\bigr)$.  Such vertical
vectors are derivatives of curves in the corresponding fiber i.e.\
derivatives of curves in the Grassmann manifold $G_\DD(T_mM)$. Thus,
the residual of two vector bundles~$\kD_i$ at~$m$ is an element of the
tangent space at the common element $(\kD_i)_m$ of the Grassmann
manifold $G_\DD(TM)$, which, by Remark~\ref{rmk:grassmann-subspace},
can be regarded as an element of
$\hom\bigl((\kD_i)_m,T_mM/(\kD_i)_m\bigr)$.

%
\section{Equations}
\label{section:equations}
%
%
Computing with the order notation on manifolds might require the
determining the contact or residual of the solutions of two implicit
equations with a given contact or residual. A most basic result that
enables this sort of argument is
Proposition~\ref{prp:f_inverse_order}, which guarantees the contact of
two inverse mappings, given the contact of two diffeomorphisms.

\begin{proposition}\label{prp:f_inverse_order}
Let $(M,h_M)$ and $(N,h_N)$ be manifolds, and let $f_i\colon M\to N$
be $C^k$~diffeomorphisms, $k\ge1$, be such that $f_i$ maps
$h_M^{-1}(0)$ into $h_N^{-1}(0)$, $i=1,2$.  Then $f_2=f_1+O(h_M^r)$
implies $f^{-1}_2=f^{-1}_1+O(h_N^r)$.
\end{proposition}

\begin{proof}
Suppose $l$ is such that $g_2\circ f_2=g_1\circ f_1+O(h_M^l)$. This is
true for $l=1$, because $g_2\circ f_2=g_1\circ f_1$ on
$h_M^{-1}(0)$. Taking the residuals of $f_i\circ g_i=\!1$, one obtains
\[
0=\res^l(f_2\circ g_2, f_1\circ g_1)(n)=\dot g_2(n)^l\,\res^l(f_2,f_1)(m)
  +T_nf_1\,\res^l(g_2,g_1)(n).
\]
Thus $\res^l(g_2,g_1)(n)=0$ if $\res^l(f_2,f_1)(m)=0$
i.e.~$g_2=g_1+O(h_N^l)$ if~$f_2=f_1+O(h_M^l)$, which inductively
gives $g_2=g_1+O(h_N^r)$.
\end{proof}

Another requirement is to semiglobally construct mappings from
graphs. Proposition~\ref{prp:mapping-from-graph} uses the semiglobal
inverse function theorem to provide such a result for a perturbation
of an identity mapping.

\begin{proposition}\label{prp:mapping-from-graph} 
Let $M$ and $(N,h_N)$ be manifolds. Let $\gamma\colon U\subseteq N\to
M\times M$ be $C^k$, $k\ge 1$. Suppose that $h_N^{-1}(0)\subseteq U$
and $\left.\gamma\right|h_N^{-1}(0)$ is a diffeomorphism
to~$\Delta(M\times M)$. Then there are neighborhoods $\tilde
U\subseteq U$ of~$h_N^{-1}(0)$ and $V\subseteq M\times\mathbb{R}$
of~$M\times\sset0$ such that, for all $(m,h)\in V$, there is a
unique~$\tilde m\in M$ such that, for some $n\in\tilde U$,
$\gamma(n)=(m,\tilde m)$ and~$h_N(n)=h$.  The map $f_\gamma\colon V\to
M$ defined by $f_\gamma(m,h)\equiv\tilde m$ is~$C^k$.
\end{proposition}

\begin{proof}
Let $\pi_1$ and~$\pi_2$ be the projections on $M\times M$
i.e.~$\pi_i(m_1,m_2)\equiv m_i$, $i=1,2$.  Define $\psi\colon U\to
M\times\RR$ by $\psi(n)\equiv\bigl((\pi_1 \circ
\gamma)(n),h_N(n)\bigr)$.  The map $\psi$ is a diffeomorphism from
$h_N^{-1}(0)$ to~$M\times\sset0$ and, by the inverse function
theorem, is a local diffeomorphism at each point of~$h_N^{-1}(0)$. By
Lemma~\ref{thm:semiglobal-inverse-fn}, $\psi$~is a diffeomorphism from
a neighborhood~$\tilde{U}\subseteq U$ of~$h_N^{-1}(0)$ to a
neighborhood~$V$ of~$M\times\sset0$.

If $(m,h)\in V$, then let $n\in\tilde U$ be such that $\psi(n)=(m,h)$,
and define $\tilde m\equiv\pi_2\bigl(\gamma(n)\bigr)$, so that
$f_\gamma(m,h)\equiv\tilde m=(\pi_2\circ\gamma\circ\psi^{-1})(m,h)$.
From $\psi(n)=(m,h)$ follows
$\bigl(\pi_1\bigl(\gamma(n)\bigr),h_N(n)\bigr)=(m,h)$ so
$\gamma(n)=(m,\tilde m)$ and $h_N(n)=h$, which are the required
properties of~$\tilde m$.  If there is another such,
say~$\tilde{m}^\prime$, then there would have to be
an~$n^\prime\in\tilde U$ such that $\gamma(n^\prime)=(m,\tilde
m^\prime)$ and~$\h_N(n^\prime)=h$, so $\psi(n^\prime)=(m,h)=\psi(n)$
which, since $\psi$ is a diffeomorphism, implies $n=n^\prime$. Hence
$(m,\tilde m)=\gamma(n)=\gamma(n^\prime)=(m,\tilde m^\prime)$, so
$\tilde m=\tilde m^\prime$.
\end{proof}

Proposition~\ref{prp:mapping-from-graph} establishes the contact of
the mappings constructed from graphs is equal to the contact of the
graphs. Further, the mappings have one higher contact if there is
present a symmetry condition for the \emph{residuals} of the graphs.

\begin{proposition}\label{prp:mapping-from-graph-contact-order}
Let $(M,h_M)$ and $(N,h_N)$ be manifolds and $\gamma_i$ and~$f_i$ be
as in Proposition~\ref{prp:mapping-from-graph}. Then
$f_{\gamma_2}=f_{\gamma_1}+O(h_M^r)$
if~$\gamma_2=\gamma_1+O(h_N^r)$. If~$\res^r(\gamma_2,\gamma_1)$ is
symmetric i.e.\ $\delta\gamma^1(n)=\delta\gamma^2(n)$ for all~$n\in
N$, where $\res^r(\gamma_2,\gamma_1)=(\delta\gamma^1,\delta\gamma^2)$,
then $f_{\gamma_2}=f_{\gamma_1}+O(h^{r+1})$.
\end{proposition}

\begin{proof}
Assume the context and notations of the proof of
Proposition~\ref{prp:mapping-from-graph}. Since
$f_{\gamma_i}=\pi_3\circ\gamma_i\circ\psi_i^{-1}$, where
$\psi_i=(\gamma_i,h_N)$,
Propositions~\ref{prp:g_comp_f_order} and~\ref{prp:f_inverse_order}
imply $f_{\gamma_2}=f_{\gamma_1}+O(h_M^r)$
if~$\gamma_2=\gamma_1+O(h^r)$. Then
\[
\pi_2\circ\gamma_i=f_{\gamma_i}\circ(\pi_1\circ\gamma_i,h_N),
\] 
so, taking the residuals of this equation at~$n\in h_N^{-1}(0)$, and
setting~$m\equiv\pi_1\bigl(\gamma_i(n)\bigr)$, gives
\[
&T_{(m,m)}\pi_2\,\res^r(\gamma_2,\gamma_1)(n)\\
&\qquad
=\res^r\bigl(f_{\gamma_2}\circ(\pi_1\circ\gamma_2,h_N),f_{\gamma_1}\circ(
  \pi_1\circ\gamma_1,h_N)\bigr)(n)\\
&\qquad=\res^r(f_{\gamma_2},f_{\gamma_1})(m,0)
+T_{(m,0)}f_{\gamma_1}\res^r\bigl((\pi_1\circ\gamma_2,h_N),
(\pi_1\circ\gamma_1,h_N)\bigr)(n)\\
&\qquad=\res^r(f_{\gamma_2},f_{\gamma_1})(m,0)
+T_{(m,0)}f_{\gamma_1}\bigl(T\pi_1\res^r(\gamma_2,\gamma_1),0\bigr)(n).
\]
Also, $f_{\gamma_1}\bigl(\pi_1(m_1,m_2),0\bigr)=m_1$ for all $m_1\in
M$, so the last term of the equation immediately above
is~$T_{(m,m)}\pi_1\bigl(\res^r(\gamma_2,\gamma_1)(n)\bigr)$, and hence
\[
\res^r(f_{\gamma_2},f_{\gamma_1})(m,0)=
T_{(m,m)}\pi_2\res^r(\gamma_2,\gamma_1)(n)
-T_{(m,m)}\pi_1\res^r(\gamma_2,\gamma_1)(n),
\]
which is zero if~$\res^r(\gamma_2,\gamma_1)(n)$ is symmetric.
\end{proof}

%
\section{Skew critical problems}
\label{section:skew-problems}
%
%
Theorem~\ref{thm:gamma_order} is a main objective of this work.  It
uses the infrastructure we have developed to show that the contact of
solutions of nondegenerate skew critical problems is the same as the
contact of their data. Moreover, the residuals of the solutions are
determined geometrically through the residuals of the data.

\begin{theorem}\label{thm:gamma_order} 
Let $(M,h_M)$ and $(N,h_N)$ be manifolds and suppose $\alpha^i$,
$g_i$, $\gamma_i$ and $\kD_i$, $i=1,2$ are as in
Theorem~\ref{theorem:semiglobal-skew-critical-point} and $M_0\subseteq
h_M^{-1}(0)$, $N_0\subseteq h_N^{-1}(0)$. If
$\alpha^2=\alpha^1+O(h_M^r)$, $g_2=g_1+O(h_M^r)$, and
$\kD_2=\kD_1+O(h_M^r)$, then $\gamma_2=\gamma_1+O(h_N^r)$.
\end{theorem}

\begin{proof}
It suffices to consider the local setup at $x=x_c$, where
\begin{enumerate}
\item $x_c\in U\subseteq\EE$ and $V\subseteq\FF$ are open
in Banach spaces $\EE$ and $\FF$, respectively, and $g(x_c)=y_c$;
\item the fiber of $\kD_i$ at $x$ is the graph
$\set{e+\Delta_i(x)e}{e\in\DD}$, where $\EE=\DD\oplus\DD^\perp$,
$\Delta_i(x)\colon\DD\to\DD^\perp$, and $\Delta_i(0)=0$;
\item $\alpha^i\colon U\to\EE^*$.
\end{enumerate}
In this setup, $x=\gamma_i(y)$ are determined by the equations
$F_i(x)=(0,y)$ such that $F_i\colon U\to\DD^*\times V$ is
defined by
\[
F_i(x)\equiv\bigl(\alpha_{\Delta_i}(x),g_i(x)\bigr),\quad
  \alpha_{\Delta_i}\equiv\alpha^i(x)\circ\bigl(\iota_{\DD\EE}
  +\Delta_i(x)\bigr),
\]
where $\iota_{\DD\EE}$ is the inclusion of $\DD$ into $\EE$.  The
domain of the $F_i$ has the local representative $h_U$ of $h_M$, and
the codomain of the $F_i$ has the function $h_{\DD^*\times
V}(\alpha,y)\equiv h_V(y)$ where $h_V$ locally represents $h_N$.
$F_2=F_1+O(h_U^r)$ since $\alpha^1=\alpha^2+O(h_U^r)$ and
$\Delta_1=\Delta_2+O(h_U^r)$, and since composition of linear maps is
continuous and bilinear. Also, $x_c$ is a nondegenerate skew critical
point for both problems corresponding to $i=1,2$, so each $F_i$ is a
local diffeomorphism at $x_c$. From
Proposition~\ref{prp:f_inverse_order}, and near $(0,y_c)$,
$F_1^{-1}=F_2^{-1}+O(h_{\DD^*\times V}^r)$, which from
$\gamma_i(y)=F_i^{-1}(0,y)$ implies $\gamma_1=\gamma_2+O(h_V^r)$, as
required.
\end{proof}

In the context of Theorem~\ref{thm:gamma_order}, we will need to know
that the residuals of the solutions $\gamma_i$ depend only on the
residuals of $\alpha^i$, $\kD_i$, and $g_i$.  For this, it suffices to
show that, given a skew critical point $m_c\in M_c$ at $n_c\in N_0$,
$u_c=\res^r(\gamma_2,\gamma_1)(n_c)$ is the unique solution of
$F(u)=0$ subject to the constraint $G(u)=0$, where $F\colon
T_{m_c}M\rightarrow\kD_{m_c}^*$ is defined by
\[[eq:residual_F]
F(u)&=\dot\gamma_2(n_c)d_{\kD}\alpha^1(m_c)(u)\\
&\qquad\mbox{}+\res^r(\alpha^2,\alpha^1)(m_c)
  +\bar\alpha^1(m_c)\circ\res^r(\kD_2,\kD_1)(m_c)
\]
and $G\colon T_{m_c}M\rightarrow T_{n_c}N$ is defined by
\[[eq:residual_G]
G(u)=Tg_1(m_c)u+\res^r(g_2,g_1)(m_c).
\]
Here $\alpha^1(m_c)$ annihilates $(\kD_1)_{m_c}$ and so descends to
$\bar\alpha^1(m_c)$ in the quotient $T_{m_c}M/(\kD_1)_{m_c}$. If
$r\ge2$ then the index~$1$ occurring asymmetrically
in~\eqref{eq:residual_F} and~\eqref{eq:residual_G}, such as in the
fragment $d_{\kD}\alpha^1(m_c)$, can be replaced by the index~$2$
because the data of the skew critical problems are assumed to match to
order.  To show~\eqref{eq:residual_F} and~\eqref{eq:residual_G}, note
that, in the local setup, a vector field in $\kD_i$ extending any
$e\in\DD$ is available as $x\mapsto e+\Delta_i(x)e$, and so
\[
d_{\kD_1}\alpha^i(x_c)(u,e)&=\left.\frac{d}{dt}\right|_{t=0}
  \langle\alpha^i(x_c+ut),e+\Delta_i(x_c+ut)e\rangle
=\langle D\alpha_{\Delta_i}(x_c)u,e\rangle.
\]
Since $\alpha_{\Delta_i}\bigl(\gamma_i(y)\bigr)=0$, the residuals of
this for $i=1,2$ are zero, so
\[
0&=\dot\gamma_2(y_c)\res^r(\alpha_{\Delta_2},\alpha_{\Delta_1})(x_c)
  +D\alpha_{\Delta_1}(x_c)\res^r(\gamma_2,\gamma_1)(y_c)\\
&=\dot\gamma_2(y_c)\res^r(\alpha^2,\alpha^1)(x_c)
  +\alpha^1(x_c)\circ\res^r(\Delta_2,\Delta_1)(x_c)\\
&\qquad\qquad\mbox{}+d_{\kD_1}\alpha^1(x_c)^\flat\res^r(\gamma_2,\gamma_1)(y_c),
\]
which is the local version of~\eqref{eq:residual_F}. The
constraint~\eqref{eq:residual_G} follows from the
equation~$\res^r(g_2\circ\gamma_2,g_1\circ\gamma_1)(n_c)=0$, since
$g_1\circ\gamma_1(n)=g_2\circ\gamma_2(n)=n$.

Suppose one has skew critical problems as in
Theorem~\ref{thm:gamma_order}, where the \emph{unperturbed}~problem is
equivariant under the action of a Lie group.  For the application we
have in mind, $\kG$~is \emph{not}~a~symmetry group of the full
critical  problem: only~the~\emph{residuals} of the unperturbed
problem are equivariant. Then
Proposition~\ref{prp-equivariance-critical-point-residual} below shows
that the residuals of the solutions are equivariant.  Recall that, if
a Lie group acts on a manifold~$M$, then it acts by lifts on~$TM$
and~$T^*M$, and also in the obvious way on the vertical bundles
of~$TM$ and~$T^*M$, and on any Grassmann bundle of~$TM$.

\begin{proposition}\label{prp-equivariance-critical-point-residual}
Let $(M,h_M)$ and~$(N,h_N)$ be manifolds suppose $\alpha^i$, $g_i$,
$\gamma_i$, and $\kD_i$,~$i=1,2$ are as in
Theorem~\ref{thm:gamma_order}. Suppose that a Lie group~$\kG$ acts
on~$M$ and~$N$, and
\begin{enumerate}
\item $g_i\colon(M,h_M)\rightarrow(N,h_N)$ i.e.\ $h_N\circ g_i=h_M$;
\item 
$\kD_i|M_0$ are tangent to $h_M^{-1}(0)$ and are invariant,
$\alpha^i|T\bigl(h_M^{-1}(0)\bigr)$ are invariant, and $g_i|h_M^{-1}(0)$ are
equivariant;
\item 
$\res^r(\alpha^2,\alpha^1)$, $\res^r(g_2,g_1)$, and
$\res^r(\gamma_2,\gamma_1)$ are equivariant.
\end{enumerate}
Then $h_M\circ\gamma_i=h_N$, and $\res(\gamma_2,\gamma_1)\colon
h_N^{-1}(0)\rightarrow T\bigl(h_M^{-1}(0)\bigr)$ is equivariant.
\end{proposition}

\begin{proof}
Since $g_i\circ\gamma_i(n)=n$,
\[h_N(n)=h_N\bigl(g_i\circ\gamma_i(n)\bigr)=h_M\circ\gamma_i(n)
\]
and $h_M\circ\gamma_i=h_N$ follows.  Fix~$\tau\in\kG$ and let
$\tilde\alpha^i=\tau^*\alpha$, $\tilde\kD_i=\tau\kD_i$, and $\tilde
g_i=\tau^*g_i$, where $\tau^*$ denotes pull-back by $m\mapsto\tau m$.
Note that the maps $\tilde\gamma_i\equiv\tau\gamma_i$ give the skew
critical points of $(\tilde\alpha^i,\tilde\kD_i,\tilde g_i)$.
By~\eqref{eq:residual_F} and~\eqref{eq:residual_G}, the residuals
$\res^r(\gamma_2,\gamma_1)$ and
$\res^r(\tilde\gamma_2,\tilde\gamma_1)$ are determined by the
residuals of the data in the skew problems $(\alpha^i,\kD_i,g_i)$ and
$(\tilde\alpha^i,\tilde\kD_i,\tilde g_i)$, respectively. So the
assumed equivariance of the data residuals implies that the residuals
$\res^r(\gamma_2,\gamma_2)$
and~$\res^r(\tilde\gamma_2,\tilde\gamma_1)$ are equal, and
\begin{equation*}
\tau\res^r(\gamma_2,\gamma_1)(n)
=\res^r(\tau\gamma_2,\tau\gamma_1)(\tau n)\\
=\res^r(\tilde\gamma_2,\tilde\gamma_1)(\tau n)\\
=\res^r(\gamma_2,\gamma_1)(\tau n),
\end{equation*}
as required.
\end{proof}
\bibliographystyle{plain}

\end{document}

%% file: defs.tex
\usepackage{amsmath,amssymb,amsfonts,latexsym,amsthm,bm}
\input{ClAuDiA.tex}

\newcommand{\res}{\operatorname{res}}
\newcommand{\ann}{\operatorname{ann}}
\newcommand{\cl}{\operatorname{cl}}
\newcommand{\St}{\operatorname{St}}
\newcommand{\GL}{\operatorname{GL}}
\newcommand{\vrt}{\operatorname{vert}}
\newcommand{\hor}{\operatorname{hor}}

\numberwithin{equation}{section}

\theoremstyle{plain}
\newtheorem{theorem}{Theorem}[section]
\newtheorem{proposition}[theorem]{Proposition}
\newtheorem{lemma}[theorem]{Lemma}

\theoremstyle{definition}
\newtheorem{definition}[theorem]{Definition}

\theoremstyle{remark}

\newtheorem{remark}[theorem]{Remark}

%% file: ClAuDiA.tex
\newcommand{\dbt}[1]{\ifcase#1{}\or.5\or1\or1.5\or2\or2.5\or3\or3.5\or4\or4.5\fi}

\renewcommand{\-}[1]{\mskip-\dbt#1mu}
\renewcommand\![1]{{\bm{{#1}}}}

\def\ClAuDiAeatblanks{\global\@ignoretrue}
\def\ClAuDiAstop{\relax}
\def\ClAuDiAsplit #1#2\stop{\string #1}
\def\[#1\]{\protect{\if\ClAuDiAsplit #1\stop[%
{\ClAuDiAwithlabel #1\ClAuDiAstop}\else%
{\ClAuDiAwithoutlabel #1\ClAuDiAstop}\fi}\ignorespaces}%
\def\ClAuDiAwithlabel[#1]#2\ClAuDiAstop{\protect{\begin{equation}\label{#1}\begin{split}#2\end{split}\end{equation}}}
\def\ClAuDiAwithoutlabel#1\ClAuDiAstop{\protect{\begin{equation*}\begin{split}#1\end{split}\end{equation*}}}

\newcommand{\set}[2]{\bigl\{\mskip1mu #1:#2\mskip1mu\bigr\}}
\newcommand{\bigset}[2]{\bigl\{\mskip1mu #1:#2\mskip1mu\bigr\}}

\newcommand{\sset}[1]{\bigl\{\mskip1mu#1\mskip1mu\bigr\}}

\newcommand{\defemph}[1]{{\sl #1}}

\let\temp=\colon\def\colon{\temp\mathopen{}}

\def\kD{{\mathcal D}}

\def\kG{{\mathcal G}}

\def\BB{{\mathbb B}}

\def\DD{{\mathbb D}}
\def\EE{{\mathbb E}}
\def\FF{{\mathbb F}}
\def\GG{{\mathbb G}}

\def\KK{{\mathbb K}}

\def\RR{{\mathbb R}}

\def\E{{\mathrm E}}

\def\h{{\mathrm h}}
